\documentclass[12pt, reqno, a4]{amsart}

\usepackage{amsmath,amssymb, mathabx}
\usepackage{graphicx}
\usepackage[pagebackref, colorlinks = true, linkcolor = blue, urlcolor  = blue, citecolor = red]{hyperref}
\usepackage{xcolor}

\usepackage[margin=1in]{geometry}

\usepackage{subcaption} 
\usepackage{colortbl}
\usepackage{multirow}
\usepackage{multicol}

\usepackage{layout}
\usepackage{enumitem}
\usepackage{breqn}

\newcommand{\eq}[1]{\begin{align}\begin{aligned}#1\end{aligned}\end{align}}

\newcommand{\pnorm}[2]{\left\|#2\right\|_{#1}}

\newtheorem{theorem}{Theorem}[section]
\newtheorem{definition}[theorem]{Definition}
\newtheorem{proposition}[theorem]{Proposition}
\newtheorem{corollary}[theorem]{Corollary}
\newtheorem{lemma}[theorem]{Lemma}
\newtheorem{remark}[theorem]{Remark}

\newtheorem{assumption}[theorem]{Assumption}

\DeclareMathOperator*{\E}{\mathbb E} 
\DeclareMathOperator{\trace}{\mathrm{Tr}} 
\DeclareMathOperator{\wg}{Wg} 

\numberwithin{equation}{section}

\begin{document}

\title{Concentration of quantum channels with random Kraus operators via matrix Bernstein inequality}

\author{Motohisa Fukuda}

\begin{abstract}
In this study, we generate quantum channels with random Kraus operators to typically obtain almost twirling quantum channels and quantum expanders. To prove the concentration phenomena, we use matrix Bernstein's inequality. In this way, our random models do not utilize Haar-distributed unitary matrices or Gaussian matrices. Rather, as in the preceding research, we use unitary $t$-designs to generate mixed tensor-product unitary channels acting on $(\mathbb C^{d})^{\otimes t}$. %$(\mathbb C^{d})^{\otimes t}$
Although our bounds in Schatten $p$-norm are valid only for $1\leq p \leq 2$, we show that they are typically quantum $\epsilon$-twirling channels with the tail bound proportional to $1/\mathrm{poly}(d^t)$, while such bounds were previously constants. The number of required Kraus operators was also improved by powers of $\log d$ and $t$, to be proportional to $td^t\log d / \epsilon^2$.
Such random quantum channels are also typically quantum expanders, but the number of Kraus operators must grow proportionally to $t \log d$ in our case. 
Finally, a new non-unital model of super-operators generated by bounded and isotropic random Kraus operators was introduced, which can be typically rectified to yield almost randomizing quantum channels and quantum expanders.
\end{abstract}

\maketitle
\markright{\MakeUppercase{Concentration of quantum channels}}

% \begin{multicols}{2}

\section{Introduction}
Concentration inequalities bound the tail probabilities of random variables. Among such statements, the Markov inequality and the Chernoff bound were extended for random matrices in \cite{ahlswede2002strong} to solve quantum information problems, where moment-generating functions of independent random matrices were processed by the Golden-Thompson inequality. In addition, Bernstein's inequality, which asserts that the sum of independent random variables concentrates around the mean, was also generalized for matrices in \cite{oliveira2009concentration} and \cite{tropp2012user}. In particular in \cite{tropp2012user} Golden-Thompson inequality was replaced by Lieb's concavity theorem of trace-exponential map. The matrix versions of Bernstein's inequality yielded various results in other fields, for example, on the topic of matrix completion \cite{recht2011simpler}. 

Almost randomizing channels have been investigated for nearly perfect security with rather shorter shared random keys \cite{hayden2004randomizing}. Let 
\eq{\label{eq:mixed unitary}
\Phi(\rho) = \frac{1}{k}\sum_{i=1}^k U_i \rho U_i^* 
}
be a mixed unitary quantum channel, where $U_i$'s are $d \times d$ unitary matrices. In \cite{hayden2004randomizing}, $\Phi$ is defined to be an $\epsilon$-randomizing channel if
\eq{\label{eq:e-randamizing}
\|\Phi(\rho) - I/d \|_p \leq \epsilon d^{1/p-1}
}
with $\epsilon>0$ and $1\leq p \leq \infty$ for all quantum states $\rho$. 
When $U_i$'s are randomly chosen with respect to the Haar probability measure, $\Phi$ was proven to be $\epsilon$-randomizing with high probability if $k \geq C d \log d / \epsilon^2$ in \cite{hayden2004randomizing}, and later the bound for $k$ was improved to $k \geq C d / \epsilon^2$ in \cite{aubrun2009almost}. Here and below, $C>0$ is a universal constant. See \cite{hastings2009classical} for a similar result. 
In contrast, when the unitary matrices $U_i$'s are chosen with respect to an isotropic measure, which is within the scope of our paper, $\Phi$ is $\epsilon$-randomizing for $p=1$ with high probability if $k \geq d \log d /\epsilon^2$ for large enough $d$ in \cite{hayden2004randomizing}, and for $1 \leq p \leq \infty$ with probability more than one-half if $k \geq C d \log^6 d /\epsilon^2$ in \cite{aubrun2009almost}. Our results in particular show that $\Phi$ is $\epsilon$-randomizing for $1\leq p \leq 2$ with probability more than $1-1/\mathrm{poly}(d)$ if $k \geq C d \log d /\epsilon^2$; see Corollary \ref{corollary:t-twirling}.

This isotropic unitary setting, which can be thought of as unitary $1$-design, was generalized to the case of unitary $t$-design in \cite{lancien2020weak}, which is related to private broadcasting \cite{broadbent2022quantum}.
More precisely, replacing $U_i$ with $U_i^{\otimes t}$ in \eqref{eq:mixed unitary} we have
\eq{\label{eq:tensor-channel}
\Phi(\rho) = \frac{1}{k}\sum_{i=1}^k U_i^{\otimes t} \rho (U_i^*)^{\otimes t} \ .
}
Then, we call $\epsilon$-twirling the following property: for all quantum states $\rho$
\eq{\label{eq:e-twirling}
\pnorm{p}{\Phi(\rho) - \E_{U \in \mathcal U(d)}\left[U^{\otimes t} \rho (U^*)^{\otimes t}\right]} \leq \epsilon d^{t(1/p-1)}  \ .
}
Here, the map $\rho \mapsto \E_{U \in \mathcal U(d)}[U^{\otimes t} \rho (U^*)^{\otimes t}]$ is called a twirling quantum channel and the expectation is taken over $U$ distributed according to the Haar measure on $\mathcal U(d)$ the matrix group of $d \times d$ unitary matrices. See \eqref{eq:unitary channel}, \eqref{eq:twirling the average} and Definition \ref{definition:almost} for further details on those concepts. 
In \cite{lancien2020weak}, it was proved that the random channel in \eqref{eq:tensor-channel} is $\epsilon$-twirling with probability more than one-half for $1 \leq p \leq \infty$ if $k \geq C (td)^t (t \log d)^6 / \epsilon^2$.
In Corollary \ref{corollary:t-twirling} of this paper, on the other hand, restricting $p$ to $1 \leq p \leq 2$, we show that the $\epsilon$-twirling property holds with probability greater than $1-2 d^{t \left( 2- \frac{C}{6}\right)}$ if $k \geq Ctd^t \log d /\epsilon^2$, approximating a wider variety of twirling quantum channels. The proof methods in \cite{aubrun2009almost} and \cite{lancien2020weak} crucially include Dudley's inequality while ours Bernstein inequality, which confines us in the range $1 \leq p \leq 2$.

Graph expanders \cite{hoory2006expander}, viewed as linear maps on probability distribution, were generalized to define quantum expanders in \cite{ben2008quantum} and \cite{hastings2007random}. Indeed, the random mixed unitary channel defined in \eqref{eq:mixed unitary} was proved to be typically a quantum expander if $U_i$'s are chosen according to the Haar probability measure \cite{hastings2007random}; see also \cite{pisier2014quantum}. Moreover, in \cite{hastings2009classical}, a similar result turns out to hold for the tensor product mixed unitary channel in \eqref{eq:tensor-channel}. 
In contrast, we show, in Corollary \ref{cororally:expander}, that the random channel in \eqref{eq:tensor-channel} is a quantum expander with probability more than $1 - 2d^{t(2 - \frac{C}{6})}$ as long as $ k \geq C t\log d/\epsilon^2$ even if $U_i$'s are drawn from the unitary $t$-design. 
Recently in \cite{lancien2024optimal} when $U_i$'s are random with respect to the unitary $2t$-design and $k \geq (\log d)^{8 + \Delta}$, the random channel in \eqref{eq:tensor-channel} was proved to be a quantum expander with probability more than $1-d^{-1}$. However, note that the regimes of \cite{hastings2007random}, \cite{hastings2009classical} and \cite{lancien2024optimal} are different from ours. See the last half of Section \ref{sec:tensor unitary channels} and \cite{lancien2024optimal} for further details.

Finally, we will introduce completely positive super-operators defined by bounded isotropic random Kraus operators, which may not be unital.
As in Section \ref{sec:new models}, typically, they can be rectified to yield 
almost randomizing channels and quantum expanders. 
Of course, non-unitary models were already investigated; quantum expanders were constructed from random isometries induced by the Haar probability measure in \cite{gonzalez2018spectral} and the spectral gap, which constitutes the conditions for quantum expanders, was investigated with random Kraus operators of Gaussian matrices \cite{lancien2022correlation}. More generalized random models without sub-Gaussian distributions can be found in \cite{lancien2023note}. However, our new random model needs only simple conditions just enough to use Bernstein's inequality; see Assumption \ref{assumption:qe} and \eqref{eq:new model}.

The organization of the paper is as follows.
Section \ref{sec:basics} introduces basic concepts. Sections \ref{sec:mats} and \ref{sec:CP} develop the foundational notions of matrices, super-operators, quantum channels, and relevant norms along with their properties.
Twirling quantum channels and quantum expanders are presented in Sections \ref{sec:intro t-design} and \ref{sec:intro qe}, respectively, together with essential mathematical preliminaries.
Section \ref{sec:bernstein} reviews a matrix version of Bernstein’s inequality.
Section \ref{sec:concentration of CPs} explores the concentration phenomena of random quantum channels. It begins with a general concentration result for quantum channels with random Kraus operators, stated in Section \ref{sec:general concentration statement}. It is followed by detailed analysis of typical behaviors of random tensor unitary channels in Section \ref{sec:tensor unitary channels}; the first part focuses on almost twirling quantum channels, and the second on quantum expanders.
Section \ref{sec:new models} then introduces new random models of quantum channels, extending the previous discussions.
Finally, Section \ref{sec:disc} concludes with further discussions.

\section{Preliminaries}\label{sec:basics}
\subsection{Basics for matrices and super-operators} \label{sec:mats}
Let $M(m,n)$ be the linear space of $m \times n$ complex matrices and $M(n) = M(n,n)$.
The set of quantum states on $\mathbb C^n$ is defined by
\eq{
S(n) = \{\rho \in M(n) \,:\, \rho^* = \rho, \, \rho \geq 0, \, \trace \rho = 1\} \ ,
}
where $*$ is the adjoint operation. For a quantum state $\rho \in S(n)$, von Neumann entropy $H(\cdot)$ is defined by
\eq{
H(\rho) = - \sum_{i=1}^n \lambda_i \log \lambda_i \ ,
}
where $\lambda_i$'s are eigenvalues of $\rho$. 
For a matrix $X \in M(m, n)$, Schatten $p$-norm for $p \geq 1$ and  Schatten $\infty$-norm, i.e. operator norm, are respectively defined as 
\eq{
\pnorm{p}{X} = \left(\trace|X|^p \right)^{\frac{1}{p}}, \qquad
\pnorm{\infty}{X} = \lim_{p \to \infty} \pnorm{p}{X} \ ,
}
where $|X|=\sqrt{X^*X}$.
Then, for a linear map $\Theta: M(n) \to M(m)$, called a super-operator, define ($q \to p$)-norm, for $1 \leq p, q \leq \infty$, as 
\eq{
\|\Theta\|_{q \to p}  = \max \left\{ \pnorm{p}{\Theta(X)} \,:\, X \in M(n) \text{ s.t. } \pnorm{q}{X}  \leq 1 \right\} \ . 
}

Now, we introduce a lemma, which translates consequences of Bernstein's inequality to statements for quantum channels. 
\begin{lemma}\label{lemma:bound}
For a super-operator $\Theta: M(n) \to M(m)$, 
\eq{\label{eq:bounding norm}
\max_{\rho \in S(n)}  \pnorm{2}{\Theta(\rho) }\leq \pnorm{2\to 2}{\Theta} \ .
}
\end{lemma}
\begin{proof}
The inclusion $S(n) \subseteq \{X \in M(n) \,:\,  \pnorm{2}{X} \leq 1\}$
shows the claim. 
\end{proof}
Note that there is an optimal matrix $\rho$ for the LHS of \eqref{eq:bounding norm} which saturates the bound $\pnorm{2}{\,\cdot\,}\leq 1$, that is, the maximum is attained by a rank-one projection.
Indeed, a spectral decomposition of a quantum state $\rho$ can be expressed as
\eq{
\rho = \sum_{i=1}^n \pi_i |v_i \rangle \langle v_i| \ .
}
Here, we used the bra-ket notation; $\langle v|$ is the dual of a vector $|v\rangle$.
The set of vectors $\{|v_i\rangle \}_{i=1}^n$ forms an orthonormal basis and $\{\pi_i\}_{i=1}^n$ a probability distribution. Then, 
\eq{
\pnorm{2}{\Theta(\rho)} 
&= \pnorm{2}{\sum_{i=1}^n \pi_i\Theta( |v_i \rangle \langle v_i| )} \leq \sum_{i=1}^n \pi_i\pnorm{2}{\Theta( |v_i \rangle \langle v_i| )} \\
&\leq \max \{ \pnorm{2}{\Theta( |v \rangle \langle v| )} \,:\, \pnorm{\text{euc}}{|v\rangle} = 1\} \ . 
}
Here, $\pnorm{\text{euc}}{\,\cdot\,}$ is the Euclidean norm of vectors. Hence $\pnorm{\text{euc}}{|v\rangle} = 1$ implies that 
$|v\rangle\langle v|$ is a rank-one projection, i.e. a quantum state, and $\pnorm{2}{|v \rangle \langle v|} = 1$.
However, the bound \eqref{eq:bounding norm} is not tight in general. Unlike the LHS, $\pnorm{2}{\Theta(\cdot)}: \{X \in M(n) : \pnorm{2}{X}\leq 1\} \to [0, \infty)$ lacks a similar property. If $\Theta$ is a quantum channel, a Hermitian input can attain $\pnorm{2\to 2}{\Theta}$ \cite{watrous2005notes}, but when the channel is noisy, mixed inputs are likely to give larger output $2$-norm. Even worse, $\Theta$ will be the difference of two super-operators in our manuscript so that there may even be no Hermitian inputs achieving $\pnorm{2\to 2}{\Theta}$.

\subsection{Completely positive super-operators in Kraus representation}\label{sec:CP}
In this subsection, we define positive and completely positive (CP) super-operators, and then quantum channels. More details can be found for example in \cite{watrous2018theory} \cite{wilde2013quantum} as well as many other standard quantum information textbooks. 

\begin{definition}
For a super-operator $\Phi: M(n) \to M(m)$,
\begin{enumerate}
\item $\Phi$ is positive if the following statement is true:
\eq{
\rho \geq 0 \quad \Longrightarrow \quad  \Phi(\rho) \geq 0 \ .
}
\item $\Phi$ is CP if $\Phi \otimes 1_{M(\ell)}$ is positive for any $\ell \in \mathbb N$, where $1_{M(\ell)}$ is the identity on $M(\ell)$.
\end{enumerate}
\end{definition}

Here is a well-known fact:
\begin{proposition}
A super-operator $\Phi: M(n) \to M(m)$ is CP if and only if
$\Phi$ can be written in Kraus form:
\eq{
\label{eq:kraus}
\Phi(X) = \sum_{i=1}^k A_i X A_i^*    
}
for some $A_1, \ldots, A_k \in M(m, n)$. 
\end{proposition}
\begin{definition}
For a CP super-operator $\Phi: M(n) \to M(m)$, we can define the following two conditions dual to each other. The notations in \eqref{eq:kraus} are used below. 
\begin{enumerate}
\item $\Phi$ is called a quantum channel if it preserves trace, i.e. 
for any $X \in M(n)$
\eq{
&\trace [\Phi (X)] = \trace [X] \ ,
\quad\text{which is equivalent to}\\
&\sum_{i=1}^k A_i^* A_i = I_n \ .
}
\item $\Phi$ is called unital if it maps the identity to the identity, i.e.
\eq{
&\Phi (I_n) = I_m \ ,
\quad\text{which is equivalent to}\quad \\
&\sum_{i=1}^k A_i A_i^* = I_m \ .
}
\end{enumerate}
\end{definition}

Now, we identify matrices as vectors and CP super-operators as matrices. To this end, first define the following linear isometric identification map for matrices:
\eq{\label{eq:id vec}
\widehat \,\, : M(n) &\to  \mathbb C^n \otimes \mathbb C^n = \mathbb C^{n^2}\\
|x \rangle \langle y| &\mapsto |x \rangle \otimes|\bar y \rangle \ .
}
Here, the isometric property implies $\pnorm{2}{X} = \pnorm{\text{euc}}{\widehat X} $ for $X \in M(m,n)$.
Then, it naturally defines an accompanying map for CP super-operators, so that \eqref{eq:kraus} is represented as 
\eq{\label{eq:id mat}
\widehat{\Phi(X)} = \widehat\Phi \widehat X = \left[\sum_{i=1}^k A_i \otimes \bar{A_i}\right] \widehat X \ .
}
Here, $\widehat\Phi \in M(m^2, n^2)$, $\widehat X \in \mathbb C^{n^2}$ and $\widehat{\Phi(X)} \in \mathbb C^{m^2}$.

The identification map in \eqref{eq:id mat} can be linearly extended to all super-operators. 
Finally, let us point out an important fact:  
\eq{
\pnorm{2 \to 2}{\Theta} = \pnorm{\infty}{\widehat \Theta}
}
for any super-operator $\Theta: M(n) \to M(m)$.

For further details on this identification, readers can consult \cite{bengtsson2017geometry},  where it is described in terms of matrix reshaping and reshuffling, and \cite{watrous2018theory}, where it is called natural representation. 

Here are the well-known statements for the spectral radius and norm of a quantum channel:
\begin{proposition}\label{proposition:spectral radius}
Let $\Phi :  M(n) \to M(n)$ be a quantum channel.
\begin{enumerate}
\item The spectral radius of $\widehat\Phi$ is $1$,
i.e. the maximum of the absolute values of the eigenvalues is $1$. 
\item The spectral norm of $\widehat \Phi$, which is $\pnorm{\infty}{\widehat\Phi}$, 
is always not less than $1$, and it is $1$ if and only if $\Phi$ is unital.
\end{enumerate}
\end{proposition}
Note that the claims hold even if complete positivity is replaced by positivity. Also, remember that a quantum channel always has a fixed quantum state. 

\subsection{\texorpdfstring{Unitary $t$-design and twirling channels}{Unitary t-design and twirling channels}}\label{sec:intro t-design}
\label{sec:intro t-design}
Let $\mathcal U(d)$ be the $d\times d$ unitary matrix group with the Haar probability measure. A finite subset $\mathcal W \subset \mathcal U(d)$ with the uniform probability is called a unitary $t$-design if 
\eq{\label{eq:t-design}
\E_{U \in \mathcal W} \left[U^{\otimes t} \otimes \bar U^{\otimes t}\right]
= 
\E_{U \in \mathcal U(d)} \left[U^{\otimes t} \otimes \bar U^{\otimes t}\right] \ .
}
Now, define the following notations
\eq{
A^1 = A, \quad A^* = A^*, \quad A^- = \bar A, \quad A^T = A^T 
}
for a matrix $A$
and extend them to simple tensor products. 
That is, for a $t$-tuple
$\gamma \in \{1, *,  -, T\}^t$, 
we write:
\eq{
A^{\otimes \gamma} = \bigotimes_{j=1}^t A^{\gamma_j} \ .
}
Then, for a $k$-tuple of (possibly different) unitary $t$-designs $\boldsymbol{W} = (\mathcal W_1, \ldots, \mathcal W_k)$ for $\mathcal U(d)$, define the following random channel: for $X \in M(d^t)$,
\eq{\label{eq:unitary channel}
\Psi^{(\boldsymbol{W}, \gamma)} (X) 
= \frac{1}{k} \sum_{i=1}^k U_i^{\otimes \gamma} \,X \, (U_i^{\otimes \gamma})^* \ ,
}
where $U_i \in \mathcal W_i$ are independent.
Also, we define the twirling channel:
\eq{\label{eq:twirling the average}
\Omega^{(d, \gamma)}(X) = \E_{U \in \mathcal U(d)} \left[U^{\otimes \gamma} X  \left(U^{\otimes \gamma}\right)^*\right]
}
for $X \in M(d^t)$, so that by \eqref{eq:t-design} we have
\eq{\label{eq:t-design-gamma}
\E_{U \in \boldsymbol{W}} \left[\widehat\Psi^{(\boldsymbol{W}, \gamma)} \right]
&= \frac{1}{k} \sum_{i=1}^k\E_{U \in \mathcal W_i} \left[U^{\otimes \gamma} \otimes \bar U^{\otimes \gamma}\right] \\
&= \E_{U \in \mathcal U(d)} \left[U^{\otimes \gamma} \otimes \bar U^{\otimes \gamma}\right] 
= \widehat\Omega^{(d, \gamma)} \ .
}
In Section \ref{sec:tensor unitary channels}, we show
$\widehat \Psi^{(\boldsymbol{W}, \gamma)}$ concentrates
around $\widehat \Omega^{(d, \gamma)}$ by Bernstein's inequality.  

For example, set $\gamma = (1^{\times t})$ and for all $X \in M(d^t)$ we have
\eq{
\Omega^{(d, \gamma)} (X) = \sum_{\alpha, \beta \in S_t}\, \trace\left[ P_{\alpha^{-1}} X  \right]\wg(\alpha^{-1}\beta, d) \, P_\beta   \ .
}
Here, $S_t$ is the symmetric group of $t$ elements and $\wg(\,\cdot\, ,d)$ is the Weingarten function \cite{collins2006integration}, which is a class function of $S_t$ defined for each $d$. 
Also $P_\alpha$ is defined as
\eq{\label{eq:permutation}
P_\alpha : (\mathbb C^d)^{\otimes t} &\to  (\mathbb C^d)^{\otimes t}\\
|v^{(1)} \rangle \otimes \cdots \otimes |v^{(t)} \rangle & \mapsto |v^{(\alpha(1))} \rangle \otimes \cdots \otimes |v^{(\alpha(t))} \rangle \ .
}
Importantly, for any $\tau \in S_t$, 
\eq{
\Omega^{(d, \gamma)} (P_\tau) = P_\tau \ .
}
As Schur-Weyl duality states, the support of $\widehat \Omega^{(d, \gamma)}$ is spanned by $\{\widehat P_\tau\}_{\tau \in S_t}$, whose dimension we denote by $r$ for now. Taking into account the fact that $\widehat \Omega^{(d, \gamma)}$ is Hermitian and idempotent, we know that $\widehat \Omega^{(d, \gamma)}$ is a rank-$r$ projection. 

\begin{definition}[Almost twirling channels and randomizing channels]
\label{definition:almost}
For $\epsilon >0$, a quantum channel $\Phi: M(d^t) \to M(d^t)$ is called an $\epsilon$-twirling channel with respect to $\gamma \in \{1, *,  -, T\}^t$ if 
\eq{\label{eq:e twirling definition}
\max_{\rho \in S(d^t)} \pnorm{p}{\Phi(\rho) - \Omega^{(d, \gamma)}} < \epsilon d^{t(1/p-1)}
}
for $1 \leq p \leq 2$. When $t=1$, it is called an $\epsilon$-randomizing channel.
\end{definition}
Note that $\Omega^{(d, (1))}(X) = \trace[X]I/d$, which is the completely randomizing quantum channel.
In this paper, we work within the rage $1 \leq p \leq 2$. However, \eqref{eq:e twirling definition} was treated for $1 \leq p \leq \infty$ in \cite{aubrun2009almost} and \cite{lancien2020weak}.

To conclude this subsection, let us introduce a lemma which relates the $(2\to 2)$-norm and uniqueness of fixed quantum state.
\begin{lemma}[A unique fixed quantum state]
\label{lemma:fixed point}
For the quantum channel $\Omega = \Omega^{(d,(1))}$ defined in \eqref{eq:twirling the average}, suppose a quantum channel $\Phi$ satisfies:
\eq{
\pnorm{2 \to 2}{\Phi - \Omega} < 1
\qquad \text{or equivalently} \qquad \pnorm{\infty}{\hat\Phi - \hat\Omega} < 1\ .
}
Then, $\Phi$ has a unique fixed quantum state. 
\end{lemma}
\begin{proof}
By Brouwer's fixed-point theorem, there is at least one fixed quantum state. Now, suppose for a contradiction that there are two different fixed quantum states $\rho, \sigma$. Then, since $\Omega(X) = \trace[X]I/d$, we have
\eq{
\pnorm{2}{\rho - \sigma} &= \pnorm{2}{\Phi(\rho) - \Phi(\sigma)}
= \pnorm{2}{\Phi(\rho - \sigma)}\\
&\leq \pnorm{2}{(\Phi- \Omega)(\rho - \sigma)} + \pnorm{2}{\Omega(\rho - \sigma)} \\
&\leq \pnorm{2 \to 2}{\Phi- \Omega} \pnorm{2}{\rho - \sigma}
}
giving a contradiction. 
\end{proof}

\subsection{Quantum expanders} \label{sec:intro qe}
Following \cite{hastings2007random} and \cite{hastings2009classical}, we adopt the following definition of quantum expanders.
\begin{definition}[quantum expanders]\label{definition:quantum expander}
A sequence of quantum channels $\{\Phi^{(d)}\}_{d=1}^\infty$ is called a quantum $(1-\epsilon)$-expander if the following individual sets of conditions are satisfied.
\begin{enumerate}
\item When $\Phi^{(d)}$ is in the form of \eqref{eq:kraus} with $m=n=d$,
\begin{enumerate}[label=(\roman*)]
\item \label{condition:small env}
Let $k(d)$ be the number of Kraus operators of $\Phi^{(d)}$, and then $k(d)/d^2 \to  0$. 
\item \label{condition:gap}
Let $\lambda_2$ be the second largest eigenvalue of $\Phi^{(d)}$ in modulus, then $|\lambda_2| < \epsilon$. 
\item \label{condition:noisy}
$\Phi^{(d)}$ is unital or $\Phi^{(d)}$ has a unique fixed quantum state and its von Neumann entropy diverges. 
\end{enumerate}

\item When $\Phi^{(d)}$ is in the form of \eqref{eq:tensor-channel},
\begin{enumerate}[label=(\roman*)]
\item \label{condition:small env tensor}
Let $k(d)$ be the number of Kraus operators of $\Phi^{(d)}$, and then $k(d)/d^{2t} \to  0$. 
\item \label{condition:gap tensor}
Let $r(d) = \mathrm{rank}(\widehat \Omega^{(d, (1^{\times t}))})$ and 
$\lambda_{r(d)+1}$ be the $(r(d)+1)$-th largest eigenvalue of $\widehat\Phi^{(d)}$ in modulus, then $|\lambda_{r(d)+1} |<\epsilon$. 
\end{enumerate}
\end{enumerate}

\end{definition}

Next, let us introduce another lemma for the spectral gaps. One can consult \cite{bhatia2013matrix} for Weyl's perturbation and majorant theorems, which we use in the proof. 
\begin{lemma}[Spectral gaps of quantum channels]
\label{lemma:spectral gap}
Let $\Omega =  \Omega^{(d, (1^{\times t}))}$ with $t \in \mathbb N$ and $r = r(d, t) = \mathrm{rank}(\widehat \Omega)$. 
Suppose a quantum channel $\Phi: M(d^t) \to M(d^t)$ satisfies: for $\delta >0$
\eq{
\pnorm{2 \to 2}{\Phi - \Omega} \leq \delta, 
\qquad \text{or equivalently} \qquad \pnorm{\infty}{\hat\Phi - \hat\Omega} \leq \delta \ . 
}
Let $\{\lambda_i(\cdot)\}_{i=1}^{d^{2t}}$ be the eigenvalues of a matrix in non-increasing order in modulus. 
\begin{enumerate}
\item Suppose $t=1$ and $\Phi$ is a quantum channel defined in \eqref{eq:kraus}. Then, $r=1$ and
\eq{
\left|\lambda_1(\widehat\Phi) \right| =  1 \qquad \text{and} \qquad 
\left|\lambda_2(\widehat\Phi) \right| \leq \delta (1 + \delta) \ .
}
\item Suppose $\Phi$ is a quantum channel defined in \eqref{eq:tensor-channel}. Then,
\eq{
\left|\lambda_r(\widehat\Phi) \right| =  1 \qquad \text{and} \qquad 
\left|\lambda_{r+1}(\widehat\Phi) \right| \leq \delta \ .
}
\end{enumerate}
\end{lemma}
\begin{proof}
Let $\{s_i(\cdot)\}_{i=1}^{d^{2t}}$ be the singular values of a matrix in non-increasing order.  
By Weyl's perturbation theorem, for any $i \in [d^{2t}]$, 
\eq{
\left|s_i(\widehat\Phi) - s_i(\widehat \Omega) \right| 
\leq \pnorm{\infty}{\widehat \Phi - \widehat \Omega}
\leq \delta \ .
}
Since $\hat\Omega$ is a projection of rank $r$ we have
\eq{
s_r(\widehat\Omega)  = \lambda_r(\widehat\Omega)  = 1  \qquad \text{and} \qquad 
s_{r+1}(\widehat\Omega)  =\lambda_{r+1}(\widehat\Omega)   = 0
}

Now, we prove the first statement. 
Since $\left|\lambda_1(\widehat\Phi) \right| =  1$ by Proposition \ref{proposition:spectral radius}, use Weyl's majorant theorem in product form:
\eq{
|\lambda_2(\widehat\Phi)| \leq \frac{s_1(\widehat \Phi)  s_2(\widehat \Phi) }{|\lambda_1(\widehat \Phi)|} \leq \delta (1 + \delta) \ .
}

Next, we prove the second statement. Notice that for any permutation matrix $P_\tau$ in \eqref{eq:permutation}, 
\eq{
\Phi(P_\tau) = P_\tau \ .
}
This means that the eigenspace of the unit eigenvalue is at least $r$-dimensional, proving $\left|\lambda_r(\widehat\Phi) \right| = 1$.
Also, since $\Phi$ is unital $s_1(\widehat \Phi)=1$.
Using Weyl's majorant theorem again, 
\eq{
|\lambda_{r+1}(\widehat\Phi)| \leq \frac{s_1(\widehat \Phi) \cdots s_{r+1}(\widehat \Phi) }{|\lambda_1(\widehat \Phi)| \cdots |\lambda_r(\widehat \Phi)|} 
\leq s_{r+1}(\widehat \Phi) \leq \delta 
}
This completes the proof. 
\end{proof}

\subsection{Matrix version of Bernstein inequality}\label{sec:bernstein}
% In this subsection, we go through a matrix version of Bernstein's inequality. 
To begin this subsection, we quote the non-Hermitian-matrix version of Bernstein's inequality from \cite{tropp2012user} and \cite{tropp2015introduction}, where one can find its proofs.
\begin{proposition}\label{proposition:bernstein}
For a sequence of centered independent random matrices $X_1, \ldots, X_k \in M(d_1, d_2)$, define the following values:
\eq{
D &= d_1 + d_2, \quad
M = \max_{i \in [k]} \operatorname*{ess\,sup}_{X_i} \pnorm{\infty}{X_i}, \\
V &= \max \left\{ \pnorm{\infty}{\sum_{i=1}^k \E \left[X_iX_i^*\right]}, \quad  \pnorm{\infty}{\sum_{i=1}^k \E \left[X_i^*X_i\right]}\right\} \ .
}
Then, for $\alpha > 0$ we have 
\eq{\label{eq:bernstein bound}
\Pr \left( \pnorm{\infty}{\sum_{i=1}^k X_i} \geq \alpha \right)
\leq D \exp\left(- \frac{\alpha^2}{2(V + M\alpha /3)}\right) \ .
}
% Note that when $X_i$'s are Hermitian, one can set $D = 2d_1 = 2d_2$.
\end{proposition} 
Note that Proposition \ref{proposition:bernstein} requires uniform bounds for the operator norms and the second moments of the independent random matrices $X_1, \ldots, X_k$, but they need not share the same distribution.

In \eqref{eq:bernstein bound}, the dimension factor $D$ can be replaced by so-called effective rank \cite{minsker2017some} for improvement with little trade-off, but we do not go in this direction. Readers can refer to \cite{tropp2015introduction} on this matter as well, where it is called intrinsic dimension. Also \cite{vershynin2018high} is a good reference for broader knowledge in concentration phenomena in high dimensional spaces.

\section{Concentration of CP super-operators} \label{sec:concentration of CPs}
\subsection{General statement of concentration}\label{sec:general concentration statement}
We apply the matrix version of Bernstein's inequality (Proposition \ref{proposition:bernstein}) to the ``matrix version of'' CP super-operators defined in \eqref{eq:id mat}.
\begin{proposition}\label{proposition:bernsetin kraus}
For a sequence of independent random matrices $A_1, \ldots, A_k \in M(m, n)$, define a random CP super-operator $\Phi: M(n) \to M(m)$ as follows: for $X \in M(n)$,
\eq{
\Phi(X) = \sum_{i=1}^k A_i X A_i^*  \ ,
}
which is equivalent to define a random matrix $\widehat\Phi \in M(m^2, n^2)$ by
\eq{
\widehat\Phi = \sum_{i=1}^k A_i \otimes \bar A_i \ .
}
Then, 
\eq{\label{eq:bernstein kraus}
\Pr \left( \pnorm{\infty}{ \widehat\Phi - \E [\widehat\Phi ]} \geq \alpha \right)
\leq D \exp\left(- \frac{\alpha^2}{2(V + M\alpha /3)}\right) \ . 
}
Here, $D = m^2 + n^2$, 
\eq{
&M = \max_{i \in [k]}  \operatorname*{ess\,sup}_{A_i} \pnorm{\infty}{A_i \otimes \bar A_i - \E \left[A_i \otimes \bar A_i \right]} \ , \\
&V = \max \Big\{ \\
&\pnorm{\infty}{\sum_{i=1}^k \E \left[A_iA_i^* \otimes \bar A_i A_i^T\right]- \E \left[A_i \otimes \bar A_i \right]\E \left[A_i^* \otimes  A_i^T\right]} , \\
& \left.\pnorm{\infty}{\sum_{i=1}^k \E \left[A_i^*A_i \otimes A_i^T\bar A_i\right] - \E \left[A_i^* \otimes A_i^T\right]\E \left[A_i \otimes \bar A_i \right]}\right\}
 . 
}
\end{proposition}
\begin{proof}
We apply Proposition \ref{proposition:bernstein} to the following centered random matrices:
\eq{
X_i = A_i \otimes \bar A_i - \E \left[A_i \otimes \bar A_i \right] \ .
}
This completes the proof. 
\end{proof}

As was discussed in the previous section, $M$ and $V$ must behave nicely for concentration, but since the $A_i$'s do not have to have the same distribution, there is room for a generalization of the results in Section \ref{sec:tensor unitary channels}, which is the next section. See Remark \ref{remark:uneven} and Assumption \ref{assumption:qe}. 

\subsection{Concentration of random mixed tensor product unitary quantum channels}
\label{sec:tensor unitary channels}
Now, based on the concepts and notations in Section \ref{sec:intro t-design}, we continue to discuss concentration phenomenon of almost twirling channels and quantum expanders. 
\begin{theorem}\label{theorem:t-design}
For $d, k, t \in \mathbb N$, a $k$-tuple of unitary $t$-designs $\boldsymbol{W} = (\mathcal 
 W_1, \ldots, \mathcal W_k)$ for $\mathcal U(d)$ and a $t$-tuple $\gamma \in \{1, *,  -, T\}^t$, 
the random quantum channel $\Psi^{(\boldsymbol{W}, \gamma)}$ defined in \eqref{eq:unitary channel} concentrates around the average $\Omega^{(d, \gamma)}$ defined in \eqref{eq:twirling the average}. For $0< \alpha \leq 1$, take $\displaystyle k = \frac{Ct \log d}{\alpha^2}$ with $C>12$ so that we have the following concentration:
\eq{
\Pr \left\{
\pnorm{\infty}{\widehat \Psi^{(\boldsymbol{W}, \gamma)} - \widehat \Omega^{(d, \gamma)}} \geq \alpha \right\}
\leq 2 d^{t \left( 2- \frac{C}{6}\right)} \ .
}
\end{theorem}

\begin{proof}
To use Proposition \ref{proposition:bernsetin kraus}, set 
$
A_i =  U_i^{\otimes \gamma} /\sqrt{k}
$ and we calculate and bound the constants in the theorem. 
First, by triangle inequality and Jensen's inequality, we have
\eq{
M = \frac{1}{k}\max_{i \in [k]} \pnorm{\infty}{U_i^{\otimes \gamma} \otimes \overline{U_i^{\otimes \gamma} } - \E_{\mathcal W_i} \left[U_i^{\otimes \gamma} \otimes \overline{U_i^{\otimes \gamma} } \right]} 
\leq \frac{2}{k} \ . 
}
Next, similarly we have
\eq{
& U_i^{\otimes \gamma} \left(U_i^{\otimes \gamma}\right)^* = I, \quad \text{ and }\quad \\
& \pnorm{\infty}{\E_{\mathcal W_i} \left[U_i^{\otimes \gamma} \otimes \overline{U_i^{\otimes \gamma} }\right] 
\E_{\mathcal W_i} \left[\left(U_i^{\otimes \gamma}\right)^* \otimes \left(U_i^{\otimes \gamma}\right)^T\right] }\\
&\leq \pnorm{\infty}{U_i^{\otimes \gamma}}^4 = 1 \ ,
}
and so forth. Hence, we get
\eq{
V \leq k \cdot \frac{2}{k^2} = \frac{2}{k} \ .
}
Therefore, 
\eq{\label{eq:first or second}
2 \left(V+ \frac{M\alpha}{3} \right) 
\leq 2 \left(\frac{2}{k} + \frac{2\alpha}{3k} \right)
\leq \frac{6}{k} \ .
}
Then, since $m = n = d^{t}$, the bound in \eqref{eq:bernstein kraus} is upper-bounded in this case by 
\eq{\label{eq:balancing tail}
2 \exp \left( 
2t \log d - \frac{\alpha^2 k}{6}
\right)
= 2 \left(d^t\right)^{2- \frac{C}{6}}
}
for $\displaystyle k = \frac{Ct \log d}{\alpha^2}$. 
The bound shows concentration when $C>12$.
\end{proof}

\begin{corollary}[Almost twirling]\label{corollary:t-twirling}
For $\displaystyle 0<\epsilon \leq d^{t/2}$, take $\displaystyle k = \frac{Ct d^t \log d}{\epsilon^2}$ with $C>12$. Then, we have for $1 \leq p \leq 2$,
\eq{
&\Pr \left\{
\max_{\rho \in S(d^t)}\pnorm{p}{\Psi^{(\boldsymbol{W}, \gamma)}(\rho) -  \Omega^{(d, \gamma)}(\rho) } \geq \epsilon d^{t \left(\frac{1}{p} - 1 \right)}  \right\}\\
&\leq 2 d^{t \left( 2- \frac{C}{6}\right)} \ .
}
This means that $\Psi^{(\boldsymbol{W}, \gamma)}$ is typically an $\epsilon$-twirling channels. 
\end{corollary} 

\begin{proof}
First, setting $\alpha = \epsilon d^{-t/2} \leq 1$ in Theorem \ref{theorem:t-design} implies 
\eq{
\pnorm{\infty}{\widehat \Psi^{(\boldsymbol{W}, \gamma)} -  \widehat\Omega^{(d, \gamma)}} < \epsilon d^{-\frac{t}{2}}
}
with probability more than $\displaystyle 1-  2 d^{t \left( 2- \frac{C}{6}\right)}$ if $\displaystyle k = \frac{Ct d^t \log d}{\epsilon^2}$. Using Lemma \ref{lemma:bound}, this bound also yields the following estimate on the $2$-norm.
\eq{
\max_{\rho \in S(d^t)}\pnorm{2}{\Psi^{(\boldsymbol{W}, \gamma)}(\rho) -  \Omega^{(d, \gamma)}(\rho)} 
< \epsilon d^{-\frac{t}{2}} \ .
}

Next, we apply the generalized H\"older's inequality. Namely, let $s_1, s_2, \ldots, s_{d^t}$ be the singular values of $\Psi^{(\boldsymbol{W}, \gamma)}(\rho) -  \Omega^{(d, \gamma)}(\rho)$, and then for $1\leq p < 2 \leq q$ such that $\frac{1}{p} = \frac{1}{2} + \frac{1}{q}$, we have, via the standard argument,
\eq{
\left( \sum_{i=1}^{d^t} (s_i \cdot 1)^p \right)^{\frac{1}{p}} 
&\leq \left( \sum_{i=1}^{d^t} s_i^2 \right)^{\frac{1}{2}} 
\left( \sum_{i=1}^{d^t}  1^q \right)^{\frac{1}{q}} \\
&= \left( \sum_{i=1}^{d^t} s_i^2 \right)^{\frac{1}{2}} \left( d^t \right)^{\frac{1}{p} - \frac{1}{2}} \ .
}
Thus, for $1 \leq p \leq 2$, we obtain
\eq{
\max_{\rho \in S(d^t)}\pnorm{p}{\Psi^{(\boldsymbol{W}, \gamma)}(\rho) -  \Omega^{(d, \gamma)}(\rho)} 
&< \epsilon d^{-\frac{t}{2}}  d^{t \left(\frac{1}{p} - \frac{1}{2} \right)} \\
&= \epsilon   d^{t \left(\frac{1}{p} - 1 \right)}  \ .
}
This completes the proof.
\end{proof}

Corollary \ref{corollary:t-twirling} states, in particular, that the random quantum channel
\eq{
\Psi^{(\boldsymbol{W}, (1^{\times t}))}(X) = \frac{1}{k} \sum_{i=1}^k U_i^{\otimes t} X (U_i^*)^{\otimes t} \ ,
}
which is the same as \eqref{eq:tensor-channel}, is typically $\epsilon$-twirling.

This random quantum channel and the case $\gamma = (1, -)$ were investigated for $1 \leq p \leq \infty$ in \cite{lancien2020weak}, where $k$ is required to grow at least on the order of $(td)^t(t \log d)^6/\epsilon^2$ and the tail bound is constant. 
This result is based on the preceding research for $t=1$ in \cite{aubrun2009almost}, where $k$ is required to grow at least on the order of $d(\log d)^6/\epsilon^2$ and the tail bound is constant. Since $\Omega^{(d,(1))}(X) = \trace[X] I/d$ for any $X \in M(d)$, the random quantum channel $\Psi^{(\boldsymbol{W}, (1))}$ was referred to as an $\epsilon$-randomizing channel in \cite{aubrun2009almost}.

The case $\gamma = (1^{\times t})$ yields another interesting consequence. Remember $\hat \Omega^{d, (1^{\times t})}$ is a projection, and consequently all its eigenvalues are $0$ or $1$. Hence, $\hat \Psi^{(\boldsymbol{W}, (1^{\times t}))}$ must have a similar eigenvalue distribution because it is close to the projection with respect to the operator norm. Here is the statement:
\begin{corollary}[$(1-\epsilon)$-expander]\label{cororally:expander}
The following random quantum channel: for $X \in M(d^t)$
\eq{
\Phi(X) = \frac{1}{k} \sum_{i=1}^k U_i^{\otimes t} X (U_i^*)^{\otimes t}  \ ,
}
where  $U_i$'s are independent (possibly different) unitary $t$-designs,
is a quantum $(1-\epsilon)$-expander with probability more than $1 - 2d^{t(2 - \frac{C}{6})}$ for $C>12$, $0 < \epsilon \leq 1$ and $\displaystyle k = \frac{C t\log d}{\epsilon^2}$. Moreover, $\epsilon$ can be chosen as small as $\omega( d^{-t} \sqrt{\log d})$. Here, $f(d)=\omega(g(d))$ means $g(d)/f(d) \to 0$ as $d \to \infty$.
\end{corollary}

\begin{proof}
First, set $\gamma = (1^{\times t})$ and $\alpha = \epsilon$ in Theorem \ref{theorem:t-design}. Then, we get the following bound as before:
\eq{
\pnorm{\infty}{\widehat \Phi -  \widehat\Omega} < \epsilon 
}
where $\Omega = \Omega^{(d, (1^{\times t}))}$. Then, using Lemma \ref{lemma:spectral gap}, the $(r+1)$-th largest eigenvalue of $\widehat \Phi$ is bounded by $\epsilon$. Clearly, $k/d^{2t} \to 0$ as $d \to \infty$ in this case. 

Next, set $\epsilon = \omega(d^{-t} \sqrt{\log d})$ so that 
\eq{
\frac{k}{d^{2t}} = \frac{Ct \log d}{\omega(\log d)} \to 0 
}
as $d \to \infty$. 
\end{proof}

By Proposition \ref{proposition:spectral radius}, we know that the largest eigenvalue of $\widehat \Phi$ is $1$. Moreover, its multiplicity is $r = \mathrm{rank}(\widehat \Omega)$; $r = t!$ if $d \geq t$. Indeed, the support of the projection $\widehat \Omega$ is spanned by $\{\hat P_\tau\}_{\tau\in S_t}$, and $\widehat \Phi \hat P_\tau = \hat P_\tau$ for all $\tau \in S_t$. Thus, Corollary \ref{cororally:expander} implies that, typically, all non-unit eigenvalues of $\widehat \Phi$ have modulus less than $\epsilon$.

In \cite{hastings2007random}, \cite{hastings2009classical} and \cite{lancien2024optimal}, the regime $\epsilon \approx \displaystyle \frac{2\sqrt{k-1}}{k}$ was investigated. Unfortunately, our results do not cover this regime. Indeed, setting $\epsilon$ to be proportional to $1/\sqrt{k}$ in the statement of Corollary \ref{cororally:expander} does not yield useful results. 

When $\gamma \neq (1^{\times t})$, the averaged matrix $\hat \Omega^{d, \gamma}$ is more involved. When $t$ is large but not too large, one can use computer programs, like \cite{fukuda2019rtni} to analyze it as a sum of tensor networks.

\begin{remark}\label{remark:uneven}
Theorem \ref{theorem:t-design}, Corollary \ref{corollary:t-twirling} and Corollary \ref{cororally:expander} can be generalized further. 
Consider the following case of uneven weights:
\eq{
\hat \Psi(X) = \sum_{i=1}^k c_i U_i^{\otimes \gamma} \otimes \overline{U_i^{\otimes \gamma}} \ , 
}
where $\sum_{i=1}^k c_i = 1$ with $0 \leq c_i \leq L/k$ for some universal constant $L$. 
When $\gamma = (1^{\times t})$, this random channel is typically an $\epsilon$-randomizing channel and a $(1-\epsilon)$-quantum expander for larger $k$. In fact, proofs for this case can be obtained by modifying the bounds of $M$ and $V$ in the above proofs. Moreover, if $c_i$'s are independently random in the same range and satisfy $\E[\sum_{i=1}^k c_i] = 1$, the above map $\Psi$ is not a quantum channel in general, but it will asymptotically yield similar results. Such a generalization for $t=1$ is investigated in the next subsection; see Assumption \ref{assumption:qe}.
\end{remark}

\subsection{New random models}\label{sec:new models}
In this subsection, we construct new random quantum channels which are, with high probability,
well-defined, and $\epsilon$-randomizing quantum channels (Definition \ref{definition:almost}) or quantum expanders (Definition \ref{definition:quantum expander}).  It is noteworthy that since our methods are based on Bernstein's inequality, we impose only second moment conditions on individual bounded random Kraus operators, i.e. they are just bounded and isotropic. 

\begin{assumption}[Distributions of independent Kraus operators]\label{assumption:qe}
Let $\{A_i\}_{i=1}^k \subseteq  M(d)$ be independent random matrices satisfying the following conditions. First, for each $A_i$
\eq{\label{eq:moment conditions}
\E [(A_i)_{x,y} (\bar A_i)_{z, w}] = \frac{1}{d}\delta(x,z) \delta(y,w) \ . 
}
In particular, it holds that $\E [A_i^*A_i] = I = \E [A_iA_i^*]$.
Next, there is a constant $L \geq 1$ such that for every $A_i$
\eq{
\pnorm{\infty}{A_i}\leq L \ .
}
We denote this joint distribution by $\boldsymbol{A}$. 
Note that the moment condition implies
\eq{
1 = \pnorm{\infty}{I} = \pnorm{\infty}{\E [A_iA_i^*]}  
\leq \E \pnorm{\infty}{A_iA_i^*} \leq L^2 \ .
}
\end{assumption}
Obviously, $A_i = U_i$ satisfy the above assumptions when $U_i$ are independently random with respect to the Haar measure or just a unitary $1$-design. Moreover, $A_i = (U_i + U_i^*)/\sqrt{2}$ yield a non-unitary example. As noted in Remark \ref{remark:uneven}, Assumption \ref{assumption:qe} can be stated under weaker conditions, giving room for scaling, but we retain the current form for simplicity. 

Define random CP super-operators: 
\eq{\label{eq:generalized mixed unitary}
\Xi^{(\boldsymbol{A})}(X) = \frac{1}{k} \sum_{i=1}^k A_i X  A_i^*  
}
where $\{A_i\}_{i=1}^k \subseteq M(d)$ satisfy Assumption \ref{assumption:qe}. 

\begin{theorem}\label{theorem:generalized cp}
The random CP super-operator $\Xi^{(\boldsymbol{A})}$ in \eqref{eq:generalized mixed unitary} concentrates around the average $\Omega = \Omega^{(1)}$ in \eqref{eq:twirling the average}. 
Set $\tilde C = 4L^4 + \frac{4}{3}L^2 + \frac{16}{3}$. 
Then, for $0 <\alpha \leq 1$ and $\displaystyle k = \frac{C\log d}{\alpha^2}$ with  $C > \tilde C$, we have 
\eq{
\Pr \left\{
\pnorm{\infty}{\widehat \Xi^{(\boldsymbol{A})} - \widehat \Omega} \geq \alpha \right\}
\leq  2 d^{2(1- C/\tilde C)} \ .
}
\end{theorem}
\begin{proof}
Since $\boldsymbol{A}$ satisfies the isotropic condition in \eqref{eq:moment conditions}, for $X \in M(d)$ we have
\eq{
\E[A X A^*] 
&= \sum_{i, j, k, \ell = 1}^d | i \rangle \langle j | \,  \E[a_{i,k}\,x_{k,\ell}\,\bar a_{j, \ell}] \\
&= \sum_{i, k=1}^d | i \rangle \langle i |  \frac{x_{k, k}}{d}= \trace[X] I/d = \Omega(X) \ .
}
The rest of the proof proceeds like in the proof of Theorem \ref{theorem:t-design}. 
Apply Proposition \ref{proposition:bernsetin kraus}, replacing $A_i$ in the statement by $A_i/\sqrt{k}$. Then, the relevant values are calculated and bounded as
\eq{
D = d^2+d^2, \qquad M \leq \frac{L^2+1}{k} , \qquad V \leq \frac{L^4+1}{k} \ .
}
Then, 
\eq{
2 \left(V+ \frac{M\alpha}{3} \right) 
\leq 2 \left(\frac{L^4+1}{k} + \frac{(L^2 +1)\alpha}{3k} \right)
\leq \frac{\tilde C}{2k} \ .
}
Hence, the bound in \eqref{eq:bernstein kraus} is upper-bounded by 
\eq{
2\exp \left( 
2 \log d - \frac{2k\alpha^2}{\tilde C}
\right)
= 2 d^{2(1- C/\tilde C)}
}
if $k$ and $C$ are chosen as in the statement of the theorem. 
\end{proof}

\begin{lemma}\label{lemma:almost invertible}
Suppose $\{A_i\}_{i=1}^k \subseteq  M(d)$ satisfy Assumption \ref{assumption:qe}. Set $\tilde C = 2L^4 + \frac{2}{3}L^2 + \frac{8}{3}$. 
Then, for $0 <\alpha \leq 1$ and $\displaystyle k = \frac{C\log d}{\alpha^2}$ with  $C > \tilde C$, we have 
\eq{
\Pr \left\{
\pnorm{\infty}{\frac{1}{k}\sum_{i=1}^k A_i^*A_i - I}  \geq \alpha \right\}
\leq 2 d^{ 1- C / \tilde C} \ .
}
\end{lemma}

\begin{proof}
Set random Hermitian matrices $X_i = (A_i^*A_i - I)/k$ and apply Proposition \ref{proposition:bernstein}. Indeed, calculate the relevant values: 
\eq{
& D = 2d, \qquad M \leq \frac{L^2 + 1}{k}  , \qquad  \\
& V \leq \frac{1}{k^2} \sum_{i=1}^k \pnorm{\infty}{
\E \left[A_i^* A_i A_i^* A_i\right] - I 
} \leq \frac{L^4+1}{k} \ .
}
Then, the proof works similarly as in the proof of Theorem \ref{theorem:generalized cp}. Since 
\eq{
2 \left(V+ \frac{M\alpha}{3} \right) 
\leq 2 \left(\frac{L^4+1}{k} + \frac{(L^2 +1)\alpha}{3k} \right)
\leq \frac{\tilde C}{k} \ .
}
and the bound in \eqref{eq:bernstein bound} is uupper-bounded by
\eq{
2\exp \left( 
\log d - \frac{k\alpha^2}{\tilde C}
\right)
= 2 d^{1- C/\tilde C} \ .
}
This complets the proof. 
\end{proof}

If the Hermitian matrix $A = \frac{1}{k}\sum_{i=1}^k A_i^*A_i$ is invertible, then, like in \cite{fukuda2022additivity}, one can rectify $\Xi^{(\boldsymbol{A})}$ to obtain a quantum channel:
\eq{\label{eq:new model}
\Psi^{(\boldsymbol{A})} (\rho) = \frac{1}{k}\sum_{i=1}^k B_i \rho B_i^* \ .
}
Here, $B_i = A_i \sqrt{A^{-1}}$ so that 
\eq{
\frac{1}{k}\sum_{i=1}^k B_i^*B_i = \sqrt{A^{-1}} \left[\frac{1}{k} \sum_{i=1}^k A_i^*A_i \right]\sqrt{A^{-1}}
= I
}
showing that $\Psi^{(\boldsymbol{A})}$ is trace-preserving. 

\begin{theorem}[Typically well-defined and noisy]\label{theorem:new model bound}
Set $\tilde C = 4L^4 + \frac{4}{3}L^2 + \frac{16}{3}$. Then, for $0 <\alpha \leq 1$ and $\displaystyle k = \frac{16 C\log d}{\alpha^2}$ with $C > \tilde C$, the random quantum channel $\Psi^{(\boldsymbol{A})}$ in \eqref{eq:new model} satisfies:
\eq{
&\Pr \left\{\text{`` $\Psi^{(\boldsymbol{A})}$ is not well-defined.'' or }
\pnorm{\infty}{\widehat \Psi^{(\boldsymbol{A})} -  \widehat \Omega } \geq \alpha \right\}\\
&\leq 2\left(1+ d^{-1}\right) d^{2(1- C/\tilde C)} \ . 
}
\end{theorem}
\begin{proof}
First let $A = \frac{1}{k}\sum_{i=1}^k A_i^*A_i$ and if $\pnorm{\infty}{A - I} < \alpha$ for $0<\alpha \leq 1/2$, then $A^{-1}$ exists and $0 \leq (1+\alpha)^{-1} \leq A^{-1} \leq (1-\alpha)^{-1}$, which implies
\eq{
&\pnorm{\infty}{\sqrt{A^{-1}} \otimes \overline{\sqrt{A^{-1}}}} \leq 2 \ ,
\qquad \text{and} \qquad  \\
&\pnorm{\infty}{\sqrt{A^{-1}} \otimes \overline{\sqrt{A^{-1}}} - I} \leq \frac{\alpha}{1 \pm \alpha} \leq 2\alpha \ .
}
Note that $\sqrt{A^{-1}}$ and hence $\sqrt{A^{-1}} \otimes \overline{\sqrt{A^{-1}}}$ are Hermitian.
Next, since 
\eq{
\widehat \Psi^{(\boldsymbol{A})} = \widehat \Xi^{(\boldsymbol{A})} \left(\sqrt{A^{-1}} \otimes \overline{\sqrt{A^{-1}}}\right)
}
we have
\eq{
&\pnorm{\infty}{\widehat \Psi^{(\boldsymbol{A})} - \widehat \Omega} \\
&\leq \pnorm{\infty}{\widehat \Xi^{(\boldsymbol{A})} \left(\sqrt{A^{-1}} \otimes \overline{\sqrt{A^{-1}}}\right) - \widehat \Omega \left(\sqrt{A^{-1}} \otimes \overline{\sqrt{A^{-1}}}\right)} \\
& \qquad + \pnorm{\infty}{\widehat \Omega \left(\sqrt{A^{-1}} \otimes \overline{\sqrt{A^{-1}}}\right) - \widehat \Omega} \\
& \leq  \pnorm{\infty}{\widehat \Xi^{(\boldsymbol{A})}  - \widehat \Omega } \pnorm{\infty}{\sqrt{A^{-1}} \otimes \overline{\sqrt{A^{-1}}}} \\
& \qquad + \pnorm{\infty}{\widehat \Omega} \pnorm{\infty}{\sqrt{A^{-1}} \otimes \overline{\sqrt{A^{-1}}} - I}\\
& \leq 2  \pnorm{\infty}{\widehat \Xi^{(\boldsymbol{A})}  - \widehat \Omega } 
+ 2\alpha \ . 
}
Therefore, replace $\alpha$ by $\alpha/4$ in Theorem \ref{theorem:generalized cp} and Lemma \ref{lemma:almost invertible} and apply union bound to prove our claim. Namely, the tail bound is obtained by
\eq{
2 d^{2(1-C/\tilde C)} + 2 d^{1-2C/\tilde C} = 2(1+d^{-1}) d^{2(1-C/\tilde C)}.
}
Here, $\tilde C$ is that in Theorem \ref{theorem:generalized cp}, which is twice as large as that in Lemma \ref{lemma:almost invertible}.
\end{proof}

\begin{corollary}
[Three regimes]\label{corollary:three regimes}
For the random quantum channel $\Psi^{(\boldsymbol{A})}$ in \eqref{eq:new model},
let $\lambda_2$ be the second largest eigenvalue of $\widehat \Psi^{(\boldsymbol{A})}$ in modulus. Then, the following individual statements hold for well-defined $\Psi^{(\boldsymbol{A})}$ with probability more than $1-2(1+d^{-1}) d^{2(1- C/\tilde C)}$ where $\tilde C = 4L^4 + \frac{4}{3}L^2 + \frac{16}{3} < C$. 
\begin{enumerate}
\item $\displaystyle k = \frac{64Cd \log d}{\epsilon^2}$ with $0< \epsilon < 1$: $\Psi^{(\boldsymbol{A})}$ is a quantum expander such that
\eq{
|\lambda_2| \leq \frac{\epsilon}{\sqrt{d}} \ ,
} and it is also an $\epsilon$-randomizing channel; $\epsilon/2$-randomizing, in fact. 
\item $\displaystyle k =\frac{16Cd}{(1-\Delta)}$ with $0<\Delta<1$: $\Psi^{(\boldsymbol{A})}$ is a quantum expander such that
\eq{
|\lambda_2| \leq 2 \sqrt{\frac{(1-\Delta)\log d}{d}} \ .
}
\item $\displaystyle k =\frac{64C\log d}{\epsilon^2}$ with $0< \epsilon \leq 1$: $\Psi^{(\boldsymbol{A})}$ has a spectral gap such that $|\lambda_2| \leq \epsilon$.
\end{enumerate}
For quantum expanders the conditions in (1) of Definition \ref{definition:quantum expander} are used. 
\end{corollary}

\begin{proof}
We apply Theorem \ref{theorem:new model bound}.
First, let $\alpha = \epsilon/2$ and $\displaystyle k =\frac{64C\log d}{\epsilon^2}$. Then, typically 
$\Psi^{(\boldsymbol{A})}$ is well-defined and 
\eq{
\pnorm{\infty}{\widehat \Psi^{(\boldsymbol{A})} -  \widehat \Omega } < \frac{\epsilon}{2} \ .
}
Hence, Lemma \ref{lemma:spectral gap} shows the last claim. Note that this condition is the most loose among the three, so $\Psi^{(\boldsymbol{A})}$ is also typically well-defined in the other two regimes. 

Next, let $\displaystyle \alpha = \sqrt{\frac{(1-\Delta)\log d}{d}}$ and then $\displaystyle k =\frac{16Cd}{(1-\Delta)}$. Hence, typically 
\eq{
&\pnorm{\infty}{\widehat \Psi^{(\boldsymbol{A})} -  \widehat \Omega } < \sqrt{\frac{(1-\Delta)\log d}{d}} \qquad \text{i.e.} \qquad  \\
&\max_{\rho \in S(d)}\pnorm{2}{\Psi^{(\boldsymbol{A})}(\rho) - I/d}^2 <  \frac{(1-\Delta)\log d}{d} \ .
}
Using a standard inequality we have
\eq{
\min_{\rho \in S(d)} H(\Psi^{(\boldsymbol{A})}(\rho)) 
&\geq \log d - d \max_{\rho \in S(d)} \pnorm{2}{\Psi^{(\boldsymbol{A})}(\rho) -  I/d}^2 \\
&\geq \Delta \log d \ . 
}
Here, we used that following fact: let $\{\lambda_i\}_{i=1}^d$  be the eigenvalues of a quantum state, then
\eq{
&\log d + \sum_{i=1}^d \lambda_i \log \lambda_i 
= \sum_{i=1}^d \lambda_i (\log d + \log \lambda_i )\\
&\leq \log \left( \sum_{i=1}^d \lambda_i (d \lambda_i) \right) 
\leq \left(\sum_{i=1}^d d \lambda_i^2 \right)- 1 \\
&= d \sum_{i=1}^d \left( \lambda_i - \frac{1}{d}\right)^2 \ ,
}
where we used the equality $\sum_{i=1}^d \lambda_i = 1$, Jensen's inequality, and the bound $\log x \leq x-1$ for $x>0$. 
Therefore, whatever the unique fixed point is (Lemma \ref{lemma:fixed point}), the entropy diverges as $d \to \infty$.
This satisfies the condition \ref{condition:noisy} for $\Psi^{(\boldsymbol{A})}$ to become a quantum expander. 
Other conditions \ref{condition:small env} and \ref{condition:gap} are clearly satisfied. Hence, second statement has been proved. 

Finally, let $\alpha = \epsilon / (2\sqrt{d})$ and then $k =64Cd\log d/\epsilon^2$, which shows the first statement in a similar way, completing the proof. 
\end{proof}

\section{Discussions}\label{sec:disc}
In this study, we used Schatten $\infty$-norm, i.e. operator norm, to measure the distance between two super-operators in matrix form. From this viewpoint, one can naturally use Bernstein's inequality for concentration phenomena. Moreover, this method makes it easier to evaluate singular values and eigenvalues of random quantum channels, in comparison with the average, applying Weyl's perturbation and majorant theorems. In this way, one can plainly discuss spectral gaps of quantum channels. However, the number of Kraus operators must grow at least proportional to the logarithm of the system dimension for quantum expanders because of the nature of Bernstein's inequality.

In Corollary \ref{corollary:three regimes}, we explored new random quantum channels which are almost randomizing channels or quantum expanders. To define this model, we imposed only two conditions on random Kraus operators (Assumption \ref{assumption:qe}), which are needed just for Bernstein's inequality to work. This is an advantage of this method, but then, of course, the tail bound cannot be exponentially small, unlike with the Haar-distributed unitary matrices. 

While applying Bernstein's inequality, our calculations for the bounds, namely $M$ and $V$ for example in Proposition \ref{proposition:bernsetin kraus}, were not tight because the loss is not very significant. Rather, it matters that the Schatten $\infty$-norm of Kraus operators is at the order of $k^{-1}$.  In this paper, random mixed unitary channels are defined with the equal weight $1/k$, but one can perturb them within the order of $k^{-1}$ without spoiling this paper's framework. 

Finally the use of Bernstein's inequality is advantageous in scaling models in terms of tensor product because it treats operator norm directly. Otherwise, one needs to make bounds corresponding to all inputs in the tensor-product input space, using $\epsilon$-nets or chaining arguments. 

\section*{Acknowledgments}
MF acknowledges JSPS KAKENHI Grant Number JP20K11667. MF appreciates the hospitality of the Centre International de Rencontres Math\'ematiques, where he received useful feedback during the workshop ``Bridges between Machine Learning and Quantum Information Science''.
MF appreciates useful comments by C\'ecilia Lancien on the first arxiv version.
MF thanks the anonymous referees of IEEE transactions on Information Theory.

%\bibliographystyle{alpha}
%\bibliography{reference}{}

\begin{thebibliography}{BASTS08}

\bibitem[Aub09]{aubrun2009almost}
Guillaume Aubrun.
\newblock On almost randomizing channels with a short kraus decomposition.
\newblock {\em Communications in mathematical physics}, 288(3):1103--1116, 2009.

\bibitem[AW02]{ahlswede2002strong}
Rudolf Ahlswede and Andreas Winter.
\newblock Strong converse for identification via quantum channels.
\newblock {\em IEEE Transactions on Information Theory}, 48(3):569--579, 2002.

\bibitem[BASTS08]{ben2008quantum}
Avraham Ben-Aroya, Oded Schwartz, and Amnon Ta-Shma.
\newblock Quantum expanders: Motivation and constructions.
\newblock In {\em 2008 23rd Annual IEEE Conference on Computational Complexity}, pages 292--303. IEEE, 2008.

\bibitem[BGGS22]{broadbent2022quantum}
Anne Broadbent, Carlos~E Gonz{\'a}lez-Guill{\'e}n, and Christine Schuknecht.
\newblock Quantum private broadcasting.
\newblock {\em Physical Review A}, 105(2):022606, 2022.

\bibitem[Bha13]{bhatia2013matrix}
Rajendra Bhatia.
\newblock {\em Matrix analysis}, volume 169.
\newblock Springer Science \& Business Media, 2013.

\bibitem[B{\.Z}17]{bengtsson2017geometry}
Ingemar Bengtsson and Karol {\.Z}yczkowski.
\newblock {\em Geometry of quantum states: an introduction to quantum entanglement}.
\newblock Cambridge university press, 2017.

\bibitem[C{\'S}06]{collins2006integration}
Beno{\^\i}t Collins and Piotr {\'S}niady.
\newblock Integration with respect to the haar measure on unitary, orthogonal and symplectic group.
\newblock {\em Communications in Mathematical Physics}, 264(3):773--795, 2006.

\bibitem[FHS22]{fukuda2022additivity}
Motohisa Fukuda, Takahiro Hasebe, and Shinya Sato.
\newblock Additivity violation of quantum channels via strong convergence to semi-circular and circular elements.
\newblock {\em Random Matrices: Theory and Applications}, 11(01):2250012, 2022.

\bibitem[FKN19]{fukuda2019rtni}
Motohisa Fukuda, Robert K{\"o}nig, and Ion Nechita.
\newblock Rtni—a symbolic integrator for haar-random tensor networks.
\newblock {\em Journal of Physics A: Mathematical and Theoretical}, 52(42):425303, 2019.

\bibitem[GGJN18]{gonzalez2018spectral}
Carlos~E Gonz{\'a}lez-Guill{\'e}n, Marius Junge, and Ion Nechita.
\newblock On the spectral gap of random quantum channels.
\newblock {\em arXiv preprint arXiv:1811.08847}, 2018.

\bibitem[Has07]{hastings2007random}
Matthew~B Hastings.
\newblock Random unitaries give quantum expanders.
\newblock {\em Physical Review A—Atomic, Molecular, and Optical Physics}, 76(3):032315, 2007.

\bibitem[HH09]{hastings2009classical}
MB~Hastings and AW~Harrow.
\newblock Classical and quantum tensor product expanders.
\newblock {\em Quantum Information \& Computation}, 9(3):336--360, 2009.

\bibitem[HLSW04]{hayden2004randomizing}
Patrick Hayden, Debbie Leung, Peter~W Shor, and Andreas Winter.
\newblock Randomizing quantum states: Constructions and applications.
\newblock {\em Communications in Mathematical Physics}, 250:371--391, 2004.

\bibitem[HLW06]{hoory2006expander}
Shlomo Hoory, Nathan Linial, and Avi Wigderson.
\newblock Expander graphs and their applications.
\newblock {\em Bulletin of the American Mathematical Society}, 43(4):439--561, 2006.

\bibitem[Lan24]{lancien2024optimal}
C{\'e}cilia Lancien.
\newblock Optimal quantum (tensor product) expanders from unitary designs.
\newblock {\em arXiv preprint arXiv:2409.17971}, 2024.

\bibitem[LM20]{lancien2020weak}
C{\'e}cilia Lancien and Christian Majenz.
\newblock Weak approximate unitary designs and applications to quantum encryption.
\newblock {\em Quantum}, 4:313, 2020.

\bibitem[LPG22]{lancien2022correlation}
C{\'e}cilia Lancien and David P{\'e}rez-Garc{\'\i}a.
\newblock Correlation length in random {MPS} and {PEPS}.
\newblock In {\em Annales Henri Poincar{\'e}}, volume~23, pages 141--222. Springer, 2022.

\bibitem[LY23]{lancien2023note}
C{\'e}cilia Lancien and Pierre Youssef.
\newblock A note on quantum expanders.
\newblock {\em arXiv preprint arXiv:2302.07772}, 2023.

\bibitem[Min17]{minsker2017some}
Stanislav Minsker.
\newblock On some extensions of bernstein’s inequality for self-adjoint operators.
\newblock {\em Statistics \& Probability Letters}, 127:111--119, 2017.

\bibitem[Oli09]{oliveira2009concentration}
Roberto~Imbuzeiro Oliveira.
\newblock Concentration of the adjacency matrix and of the laplacian in random graphs with independent edges.
\newblock {\em arXiv preprint arXiv:0911.0600}, 2009.

\bibitem[Pis14]{pisier2014quantum}
Gilles Pisier.
\newblock Quantum expanders and geometry of operator spaces.
\newblock {\em Journal of the European Mathematical Society}, 16(6):1183--1219, 2014.

\bibitem[Rec11]{recht2011simpler}
Benjamin Recht.
\newblock A simpler approach to matrix completion.
\newblock {\em Journal of Machine Learning Research}, 12(12), 2011.

\bibitem[Tro12]{tropp2012user}
Joel~A Tropp.
\newblock User-friendly tail bounds for sums of random matrices.
\newblock {\em Foundations of computational mathematics}, 12:389--434, 2012.

\bibitem[Tro15]{tropp2015introduction}
Joel~A. Tropp.
\newblock An introduction to matrix concentration inequalities.
\newblock {\em Foundations and Trends® in Machine Learning}, 8(1-2):1--230, 2015.

\bibitem[Ver18]{vershynin2018high}
Roman Vershynin.
\newblock {\em High-dimensional probability: An introduction with applications in data science}, volume~47.
\newblock Cambridge university press, 2018.

\bibitem[Wat05]{watrous2005notes}
John Watrous.
\newblock Notes on super-operator norms induced by schatten norms.
\newblock {\em Quantum Information \& Computation}, 5(1):58--68, 2005.

\bibitem[Wat18]{watrous2018theory}
John Watrous.
\newblock {\em The theory of quantum information}.
\newblock Cambridge university press, 2018.

\bibitem[Wil13]{wilde2013quantum}
Mark Wilde.
\newblock {\em Quantum information theory}.
\newblock Cambridge university press, 2013.

\end{thebibliography}

% \end{multicols}

\end{document}